%% file: main.tex
\documentclass[journal]{IEEEtran}
\renewcommand\IEEEkeywordsname{Index Terms}
\usepackage{blindtext}
\usepackage{graphicx}
\usepackage{amsmath}
\usepackage{amsfonts}
\usepackage{braket}
\usepackage{multirow}
\usepackage{subfigure}
\usepackage{cite}
\usepackage[ruled]{algorithm2e}
\usepackage[utf8]{inputenc}
\usepackage[english]{babel}
 
\usepackage{amsthm}
\ifCLASSINFOpdf
\else
\fi
\input{newcommands}

\begin{document}
%
\newtheorem{corol}[theorem]{Corollary}
\title{Frequency-Domain Decoupling  for MIMO-GFDM Spatial Multiplexing}
%
%
%

\author{Ching-Lun Tai, 
        Borching Su, 
        and~Cai Jia
}

\maketitle

\begin{abstract}
Generalized frequency division multiplexing (GFDM) is considered a non orthogonal waveform and can cause difficulties when used in the spatial multiplexing mode of a multiple-input-multiple-output (MIMO) scenario.
In this paper, a class of GFDM prototype filters, in which the GFDM system is free from inter-subcarrier interference, is investigated, enabling frequency-domain decoupling during processing at the GFDM receiver.
An efficient MIMO-GFDM detection method based on depth-first sphere decoding is subsequently proposed with this class of filters. 
Numerical results confirm a significant reduction in complexity, especially when the number of subcarriers is large, compared with existing methods.
\end{abstract}

\begin{IEEEkeywords}
Frequency-domain (FD) decoupling,
multiple-input-multiple-output generalized frequency division multiplexing (MIMO-GFDM),
spatial multiplexing (SM),
depth-first sphere decoding (DFSD), 
sorted QR decomposition (SQRD), 
symbol error rate (SER).
\end{IEEEkeywords}

%
\IEEEpeerreviewmaketitle

\section{Introduction}

Generalized frequency division multiplexing (GFDM), considered as a generalization of the conventional orthogonal frequency division multiplexing (OFDM), was studied recently as a new candidate waveform for future wireless communication systems \cite{michailow14}.
In addition to an appropriate pulse shaping filter, it features several advantages including low out-of-band (OOB) emissions and relaxed frequency synchronization requirements. 
However, most prototype filters used for GFDM make it a non orthogonal system, causing inter subcarrier interference (ICI) and inter subsymbol interference (ISI) problems. Consequently, when multiple-input-multiple-output (MIMO) scenarios are considered, the cancellation of various types of interference, that is, inter antenna inteference (IAI) along with ICI and ISI, becomes a severe complication that makes detection in MIMO-GFDM difficult.


To resolve this problem, many methods of MIMO-GFDM equalization and detection have been studied.
Linear equalizers for MIMO-GFDM, including zero-forcing (ZF)\cite{tunali15} and minimized mean square error (MMSE) \cite{zhang17}, possess a relatively high symbol error rate (SER). 
Existing MIMO-GFDM detection methods can roughly be categorized into non iterative\cite{matthe15b, matthe16d} and iterative \cite{zhang15,zhang16,matthe16a,matthe17} schemes.
To remove ICI, ISI, and IAI simultaneously, non iterative receivers \cite{matthe15b,matthe16d} possess prohibitively high complexity.
Iterative receivers \cite{matthe16a, matthe17}, by contrast, have an affordable complexity in each iteration. However, the required iterations may cause processing latency, rendering them unsuitable for critical-time applications \cite{simsek16}.


In this study, we investigate a class of prototype filters that are ICI-free and can be applied in MIMO-GFDM with low complexity and low latency.
An example is the Dirichlet filter \cite{matthe14a}, which was reported before but has not been widely used.
We extend the Dirichlet filter to a class of ICI-free prototype filters and study non-iterative receivers for MIMO-GFDM with such filters.

The remainder of this paper is organized as follows. In Section \ref{sec:system}, we introduce the MIMO-GFDM system model and existing detection methods. The proposed scheme is illustrated in Section \ref{sec:method}. Simulation results and discussion are presented in Section \ref{sec:simulation}. Finally, a conclusion is provided in Section \ref{sec:conclusion}.

\emph{Notations:} Boldfaced capital and lowercase letters denote matrices and column vectors, respectively. We use $\e\{\cdot\}$ to denote the expectation operator. 
Given a vector $\vu$, we use $[\vu]_n$ to denote the $n$th component of $\vu$, $\|\vu\|$ the $\ell_2$-norm of $\vu$, and $\diag(\vu)$ the diagonal matrix containing $\vu$ on its diagonal. 
Given a matrix $\mA$, we denote $[\mA]_{m, n}$, $\vect(\mA)$, ${\mA}^T$, and ${\mA}^H$ its ($m$, $n$)th entry (zero-based indexing), column-wise vectorization, transpose, and Hermitian transpose, respectively. 
For any matrices $\mA$ and $\mB$, we use $\mA \otimes \mB$ to denote their Kronecher product.  
We define ${\mI}_p$ to be the $p \times p$ identity matrix, $\textbf{1}_p$ the $p \times 1$ vector of ones, ${\mW}_p$ the normalized $p$-point discrete Fourier transform (DFT) matrix with ${[\mW_p]}_{m, n}=e^{-j2\pi mn/p}/\sqrt{p}, p \in \bN$, and $\gd_{kl}$ the Kronecker delta. 
For any set $\cA$, we use $|\cA|$ to denote its cardinality. 
Given matrices $\mA_l, ~\forall 0 \leq l < p, p \in \bN$, of size $m \times n$, we use \mbox{blkdiag}($\{\mA_l\}_{l=0}^{p-1}$) to denote the $pm \times pn$ block diagonal matrix whose $l$th diagonal block is $\mA_l$.
For any $A\in \bN$, we use ${\bf \Pi}_A$ to denote the $A\times A$ permutation matrix ${\bf \Pi}_A = \mx{\zv^T & 1 \\ \mI_{A-1} & \zv}$.
For any $A,B \in \bN$, we use ${\bf\Pi}_{AB}$ to denote the $AB\times AB$ permutation matrix defined as
$
{[{\bf\Pi}_{AB}]}_{mB+p,qA+n}=\delta_{m n}\delta_{p q},
\label{eq:pi_AB}
$
$ \forall m,n\in\{0,1,...,A-1\},\forall p,q\in\{0,1,...,B-1\}.$

\section{System Model for MIMO-GFDM and Problem Formulation}
\label{sec:system}
Consider a MIMO-GFDM system with $T$ transmit and $R$ receive antennas operating in spatial multiplexing (SM) mode. 
Let $\vd_t \in\mathbb{C}^D$ be the data vector at the $t$th transmit antenna, $t=1,2,...T$, which satisfies $E[\vd_t] = \zv$ and $E[\vd_t\vd_t^H] = E_s\mI_D$, where $E_s$ is the symbol energy. Given a GFDM prototype filter $\vg$ \cite{michailow14}, we denote the corresponding transmitter matrix as \cite{michailow14}
\begin{equation}
\mA=[\vg_{0,0}...\vg_{K-1,0} \quad \vg_{0,1}... \vg_{K-1,1}...\vg_{K-1,M-1}],
\label{eq:AMatrix}
\end{equation}
where $\vg_{k,m}$ pulse-shapes the $m$th subsymbol on the $k$th subcarrier of $\vd_t$, with its $n$th entry being $[\vg_{k,m}]_n=[\vg]_{\modd{n-mK}_D}e^{j2 \pi kn/K}$, $n=0, 1, ..., D-1$, $m=0,1,...,M-1,k=0,1,...,K-1$.
The frequency-domain (FD) prototype filter is defined as $\vg_f = \sqrt{D}\mW_D\vg$. Subsequently, the data vector $\vd_t$ is modulated by $\mA$. 
The modulated data vector passes through the process consisting of a cyclic prefix (CP) insertion of length $L$, a linear time-invariant (LTI) channel between the $t$th transmit and $r$th receive antennas, and a CP removal. This process can be denoted as $\mH_{r,t}$, which is a circulant channel matrix between the $t$th transmit and $r$th receive antennas, and $\vy_r$ is the received signal at the $r$th receive antenna.
The signal at the receive antennas can be expressed as \cite{matthe15b}

\begin{equation}
\underbrace{\mx{\vy_1 \\ \vdots \\ \vy_R}}_{\vy} = \underbrace{\mx{
\mH_{1,1}\mA  & \cdots & \mH_{1,T}\mA\\
\vdots & \ddots & \vdots\\
\mH_{R,1}\mA  & \cdots & \mH_{R,T}\mA\\}}_{\tilde{\mH}}
\underbrace{\mx{\vd_1  \\ \vdots \\ \vd_T}}_{\vd} + \vn,
\label{eq:y}
\end{equation}
where $\vn \sim \cC \cN (0,N_0 \mI_{RD})$ is additive white Gaussian noise (AWGN) and $\mA$ is defined as in (\ref{eq:AMatrix}).
We denote $\tilde\mH$ as the $RD\times TD$ characteristic matrix of MIMO channels, $\vd$ the transmitted data vector, and $\vy$ the received data vector. 
If we set $\mA=\mW_D^H$ in (\ref{eq:y}), then it reduced to a MIMO-OFDM system.

At the receiver, an optimal detection rule can yield the maximun likelihood (ML) solution to (\ref{eq:y}) in terms of the minimum distance:
\begin{equation}
\hat \vd = \arg \underset{\vd\in \cD}{\mbox{min}} \| \vy-\tilde{\mH} \vd \|^2,
\label{eq:d_hat}
\end{equation}
where $\cD$ is the set consisting of all possible transmit symbol vectors, restricted by the constellation set. 
However, the huge size of $\cD$ proves that an exhaustive search would be infeasible, as suggested in (\ref{eq:d_hat}). 
In \cite{matthe15b,matthe16d}, a near-ML solution of the problem (\ref{eq:d_hat}) was proposed with MMSE sorted QR decomposition (SQRD) \cite{wubben03} of the matrix $\tilde\mH$.

Detection is performed by combining depth-first sphere decoding (DFSD) \cite{damen00} of groups of symbols and successive interference cancellation (SIC) between groups in series. 
However, this process is inefficient because of the prohibitively high complexity of MMSE-SQRD of $\tilde{\mH}$ and severe error propagation accompanying SIC when numerous subcarriers are employed because of the lack of FD decoupling.

\section{Proposed Method}


\label{sec:method}
\subsection{FD Decoupling}
To address the problem of the high complexity of MMSE-SQRD and severe error propagation caused by SIC in large-scale data detection, FD decoupling for MIMO-GFDM is crucial but has not been fully investigated.
In this paper, we propose a class of prototype filters to achieve FD decoupling for MIMO-GFDM, leading to dramatically decreased SQRD complexity with improved SER performance. 
The following theorem represents the foundation of the proposed scheme.

\begin{theorem}
\label{FD_decoup}

Let $\mA$ be a GFDM matrix derived from its FD prototype filter $\vg_f$ and assume that $\vg_f$ contains at most $M$ consecutive nonzero entries (i.e., there exist $\vg_1 \in \mathbb{C}^M$ and an integer $l$, $0\leq l < D$) such that
\begin{equation}
 \vg_f = {\bf\Pi}^l_D \mx{\vg_1^T & {\bf 0}^T_{(K-1)M}}^T.
\label{eq:gf_fd_decoupling}
\end{equation}
Consequently, the matrix $\tilde{\mH}$ as defined in (\ref{eq:y}) can be decomposed into the form 
\begin{equation}
\tilde{\mH} = \mU^H \mathrm{ blkdiag}(\{\mF_k\}_{k=0}^{K-1}) \mP, \label{eq:H_block_diagonalized}
\end{equation}
where $\mU = ({\bf \Pi}_{KR} \otimes \mI_M)(\mI_R \otimes {\bf\Pi}_D^{-l}\mW_D)$, $\mP = ({\bf \Pi}_{KT} \otimes \mI_M)(\mI_T\otimes {\bf \Pi}_{KM}), $ and  $\mF_k$, $k=0,...,K-1$, are some $MR\times MT$ matrices.

\end{theorem}

\begin{proof}
With some effort, we can prove that
\begin{eqnarray*}
{\bf\Pi}_D^{-l}\mW_D \mA &=& \frac{1}{\sqrt{K}} {\bf \Pi}_{KM}(\mathrm{diag}(\vg_1){\bf\Pi}^{-1}_M \mW_M \otimes \mI_K)
\end{eqnarray*}
for any $l$. 
Consequently, noting that a circulant matrix $\mH_{r,t}$ is diagonalizable using $\mW_D$, we can prove that
{\fontsize{9}{6}
\begin{eqnarray}
 {\bf\Pi}_D^{-l}\mW_D\mH_{r,t}\mA &=& \frac{1}{\sqrt{K}}\mathrm{diag}({\bf\Pi}^{-l}_D \vh_f^{(r,t)}) \nonumber \\&& \cdot(\mI_K \otimes \mathrm{diag}(\vg_1) {\bf\Pi}^{-l}_M\mW_M) {\bf \Pi}_{KM} \nonumber  \\
 &=& \mathrm{blkdiag}(\{\mE_k^{(r,t)}\}^{K-1}_{k=0}) {\bf \Pi}_{KM} \label{eq:lemma}
\end{eqnarray}}
where 
{\fontsize{8}{6}
\begin{eqnarray}
 \mE_k^{(r,t)} =& \frac{1}{\sqrt{K}}\mathrm{diag}([\mO_{M\times kM},\mI_M,\mO_{M\times (K-1-k)M}]{\bf\Pi}^{-l}_D \vh_f^{(r,t)})\nonumber\\
& \cdot \mathrm{diag}(\vg_1) {\bf\Pi}^{-l}_M\mW_M\label{eq:Ek}
\end{eqnarray}
}
and $\vh_f^{(r,t)}$ is the DFT of the first column of $\mH_{r,t}$ 
Consequently, using (\ref{eq:lemma}), we can easily verify that
{\fontsize{8}{6}
\begin{eqnarray*}
&& (\mI_R \otimes {\bf\Pi}^{-l}_D\mW_D) \tilde{\mH} \cdot (\mI_T\otimes {\bf \Pi}_{KM})\\
&=& \mx{\mathrm{blkdiag}(\{\mE_k^{(1,1)}\}_{k=0}^{K-1}) & \cdots & \mathrm{blkdiag}(\{\mE_k^{(1,T)}\}_{k=0}^{K-1})\\
\vdots &  & \vdots \\
\mathrm{blkdiag}(\{\mE_k^{(R,1)}\}_{k=0}^{K-1})&\cdots&\mathrm{blkdiag}(\{\mE_k^{(R,T)}\}_{k=0}^{K-1})}. \\
\end{eqnarray*}}
Therefore, we can prove that
\begin{eqnarray*}
\mU\tilde{\mH}&=& ({\bf \Pi}_{KR}\otimes \mI_M)(\mI_R\otimes {\bf\Pi}^{-l}_D\mW_D)\tilde{\mH} \\
&=& \mathrm{blkdiag}(\{\mF_k\}_{k=0}^{K-1}) \underbrace{({\bf\Pi}_{KT}\otimes\mI_M)(\mI_T\otimes {\bf\Pi}_{KM})}_{\mP}\\
\end{eqnarray*}
where
{\fontsize{8}{6}
$
\mF_k = \mx{\mE_k^{(1,1)} &
\cdots & \mE_k^{(1,T)}\\
\vdots & & \vdots \\
\mE_k^{(R,1)} & \cdots & \mE_k^{(R,T)}} .   
$}
The proof of Theorem 1 is complete.

\end{proof}

Theorem \ref{FD_decoup} implies that a prototype filter whose frequency domain contains only $M$ consecutive nonzero values (i.e., satisfying (\ref{eq:gf_fd_decoupling})) would enable the MIMO-GFDM system to possess FD decoupling capability, thereby leading to cheap receiver implementation of MIMO-GFDM detection, to be elaborated later.

The Dirichlet pulse \cite{matthe14a} is a typical example of this class of prototype filters, in which $\vg_1 = {\bf 1}_M$ and $l=D-\lceil-M/2\rceil$ in (\ref{eq:gf_fd_decoupling}). 
The widely used raised-cosine (RC) filter \cite{michailow14} is not a member of this class.
It is straightforward to verify that when $M=1$, the statement in Theorem 1 reduces to the special case of MIMO-OFDM, in which a rectangular window is used as the prototype filter and each subcarrier trasmits only one subsymbol in a block.






\subsection{Proposed MIMO Detection Scheme}

We now address the MIMO-GFDM detection process.
Given the received data vector $\vy$, we first perform the operation
$ \bar\vy = \mU \vy$, which involves only $R$ parallel $D$-point fast Fourier transform and several permutations.
We consequently obtain
\begin{equation}
\label{eq:y_bar}
\bar\vy = \mathrm{blkdiag}(\{\mF_k\}_{k=0}^{K-1}) \bar\vd + \bar\vn
\end{equation}
where $\bar{\vd} = \mP\vd$ and $\bar{\vn} = \mU\vn$. 
Eq. (\ref{eq:y_bar}) is in the form of block diagonalization.
Divide the vector $\bar\vy$ into $K$ segments and denote $\bar\vy_k$ as the $k$th vector of length $RM$, $k=0,1,...,K-1$, and consequently we obtain 


\begin{equation}
\bar{\vy}_k=\mF_k\bar{\vd}_k+\bar{\vn}_k, k=0,1,...,K-1,
\label{eq:yi_bar}
\end{equation}
where $\bar\vd_k$ and $\bar\vn_k$ are the $k$th parts of $\bar\vd$ and $\bar\vn$, respectively. 
The vector $\bar\vy_k$ represents the received data from the $k$th subcarrier, which depends only on the transmitted data of the $k$th subcarrier $\bar\vd_k$ and does not suffer from ICI.

To solve the subproblems of (\ref{eq:yi_bar}), we employ the SQRD \cite{wubben01} of $\mF_k$ given by 
$\mF_k = \mQ_{k}  \mR_k \mP_k^T,$ 
where $\mQ_{k}$ is an $MR \times MT$ unitary matrix, $\mR_k$ is an $MT \times MT$ upper triangular matrix, and $\mP_k$ is an $MT \times MT$ permutation matrix, which denotes the column sorting of $\mF_k$.
Consequently, by multiplying both sides of (\ref{eq:yi_bar}) with $\mQ_{k}^H$, we obtain
$\tilde{\bar{\vy}}_k = \mR_k \tilde{\bar{\vd}}_k+\tilde{\bar{\vn}}_k,
\label{eq:ybar_tilde}
$
where $\tilde{\bar{\vy}}_k=\mQ_k^H \bar{\vy}_k$, $\tilde{\bar{\vd}}_k=\mP_k^T \bar{\vd}_k$, and $\tilde{\bar{\vn}}=\mQ_k^H \bar{\vn}_k$.
Subsequently, the ML solution to each of these $K$ subproblems is computed in parallel with DFSD, without the error propagation caused by SIC. 
Consequently, the detection complexity can be dramatically reduced and the SER performance can be expected to improve. 

\subsection{Complexity Analysis}

\begin{table}[h]
\label{tab:CC}
\caption{Computational Complexity of SQRD and SIC}
\centering 
\begin{tabular}{p{1.5cm}|p{4.5cm}|p{1cm}}
 \hline
 Scheme     & SQRD  & SIC\\
\hline\hline
     OFDM    &  $D T^2 R+DTR+(D T^2-DT)/2$ (Using SQRD\cite{wubben01})   & 0\\
     \hline
Near-ML MIMO-GFDM \cite{matthe15b}    &  $K^3 M^3 T^2 R+K^2M^2TR+(2K^3M^3T^3+3K^2M^2T^2+KMT)/6$ (Using MMSE-SQRD\cite{wubben03})  & $K^2 T^2 M^2$\\
     \hline
     Proposed    &  $K M^3 T^2 R+KM^2TR+(KM^2T^2-KMT)/2$ (Using SQRD\cite{wubben01})  & 0 \\
     \hline\hline
\end{tabular}
\end{table}

\label{subsec:complexity}
To evaluate and compare the computational complexity of the proposed detection scheme with that of OFDM and conventional GFDM implementations \cite{matthe15b}, we consider the number of complex multiplications (CMs) required to detect $KMT$ symbols at the receiver, assuming that the prototype filters of all implementations take complex values. 
For a fair comparison, the block size of $KM$ is used for both GFDM and OFDM \cite{lin10}.

The data detection process consists of SQRD and depth-first SD, with additional SIC involved only in the detection process of conventional GFDM implementations.
The method in \cite{matthe15b} requires MMSE-SQRD to improve the performance of SIC; the proposed method uses only regular SQRD \cite{wubben01} to obtain the ML solution for each subproblem in (\ref{eq:yi_bar}).
Table \ref{tab:CC} shows the computational complexity of SQRD and SIC. 
As for DFSD, GFDM and OFDM receivers require $K$ times of DFSD with size $MT$ and $KM$ times of DFSD with size $T$, respectively, to detect $KMT$ symbols.
Because no analytic solution to the computational complexity of DFSD exists, we evaluate the average complexity of the entire detection process for $KMT$ symbols through Monte Carlo simulation, as described in the next section.

\section{Simulation}
\label{sec:simulation}








In this section, numerical results are provided to compare the performances, in terms of SER and complexity, of the proposed scheme with those of OFDM and conventional GFDM implementations in SM-mode MIMO systems.
We adopt the Dirichlet filter for the proposed scheme. 
Both the Dirichlet and RC filters are included for the conventional GFDM implementations \cite{matthe15b}.



The modulation is QPSK, the symbol energy is $E_S=1$, the CP length $L=D/8$, and the roll-off factor of the RC filter is $\alpha=0.9$.
We consider two cases $(K,M)=(256,4)$ and $(512,2)$ for GFDM. 
For a fair comparison, OFDM uses the same block size, namely,  $(K,M)=(1024,1)$.
The performances are evaluated through Monte Carlo simulation with randomly generated channel realizations and independent data sets for the realizations.
The channel power delay profile is exponential from 0 to -10 dB with $L$ taps.
Each simulation plot is generated with $N_h=500$ spatially uncorrelated Rayleigh fading channel realizations.
$N_d=100$ independent data blocks are generated for performance evaluation for each channel realization.
The numbers of transmit and receive antennas are set as $T=2$ and $R=2$, respectively.

Simulation results for $(K,M)=(256,4)$ are shown in Fig. \ref{fig:K256M4T2_all}, in which we compare the proposed scheme with OFDM and the conventional GFDM implementations of the Dirichlet and RC filters. 
Figs. \ref{fig:K256M4T2} and \ref{fig:K256M4T2_Complexity} present the SER performances and complexity comparisons, respectively.
In Fig. \ref{fig:K256M4T2_Complexity}, the complexity is calculated by adding the required number of CMs at the receiver, including SQRD, SIC, and SD operations. 
SQRD and SIC is calculated using the formula given in Table \ref{tab:CC}, whereas SD is counted each time during the simulation.

Fig. \ref{fig:K256M4T2_Complexity} indicates that the proposed scheme results in significant complexity reduction (approximately $10^5$ times) compared with the conventional GFDM implementations and has a complexity of only approximately 10 times that of the OFDM.
The complexity of OFDM and conventional GFDM implementations \cite{matthe15b} is virtually constant with different SNRs because SQRD dominates; by contrast, the proposed scheme has an increase in the low-SNR region because of DFSD.
The advantage in the complexity of the proposed method compared with conventional implementations, is mainly attributable to $K^2$-time reduction in the dominating term of the SQRD operation.
When we interpret the SER performance presented in Fig. \ref{fig:K256M4T2}, we observe that the proposed scheme has the most favorable SER performance of all curves.
The advantage in SER performance may be attributable to the FD decoupling property of the proposed scheme, which avoids the SER degradation caused by SIC required by conventional implementations. 
OFDM can be considered a special case of the proposed scheme in which $M=1$, with a narrower subcarrier spacing.
It exhibits a less favorable SER performance than GFDM because of low frequency diversity.

For $(K,M)=(512,2)$, the SER and detection complexity results for all schemes are shown in Figs. \ref{fig:K512M2T2} and  \ref{fig:K512M2T2_Complexity}, respectively. 
The trends are similar to those in the previous case.
In addition, we observe that the SER performance of the RC filter severely degrades when $M=2$. 

Figs. \ref{fig:K256M4T2_all} and \ref{fig:K512M2T2_all} reveal that the proposed scheme with the Dirichlet filter requires significantly less detection complexity for GFDM than conventional implementations and achieves superior SER performance. 
In addition, the choice of $M$ affects both the SER performance and detection complexity. 
Recall that OFDM is a special case of the Dirichlet filter in which $M=1$. 
Although the detection complexity drops with OFDM (i.e., the value of $M$ is minimized), the SER performance improves with the Dirichlet filter as $M$ increases, which may be attributable to the additional frequency diversity obtained from using a larger subcarrier spacing.







\begin{figure}%
\centering
\subfigure[SER]{%
\label{fig:K256M4T2}%
\includegraphics[height=1.2in]{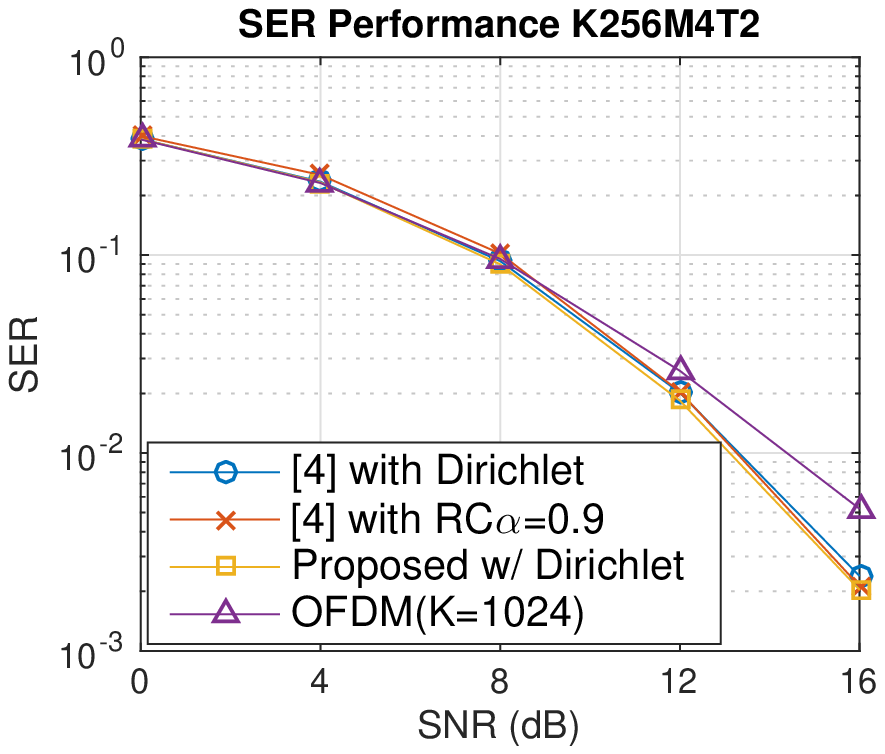}}%
~
\subfigure[Complexity]{%
\label{fig:K256M4T2_Complexity}%
\includegraphics[height=1.2in]{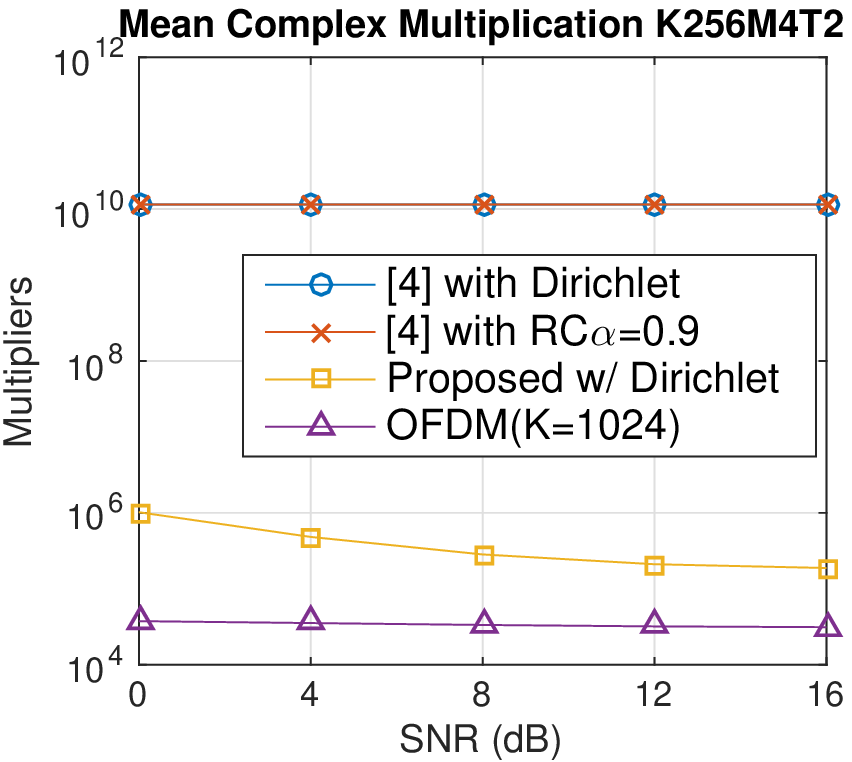}}%
\caption{Performance comparison for $K=256,M=4,T=2$}
\label{fig:K256M4T2_all}
\end{figure}

\begin{figure}%
\centering
\subfigure[SER]{%
\label{fig:K512M2T2}%
\includegraphics[height=1.2in]{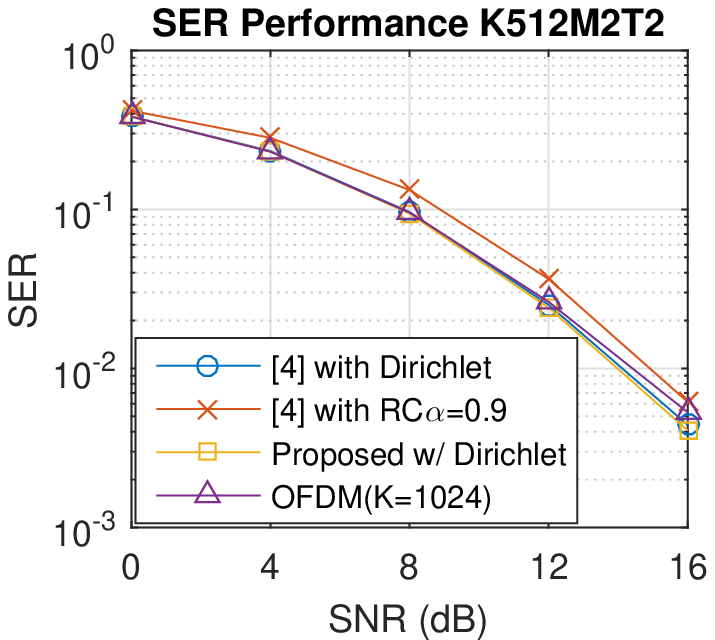}}%
\quad
\subfigure[Complexity]{%
\label{fig:K512M2T2_Complexity}%
\includegraphics[height=1.2in]{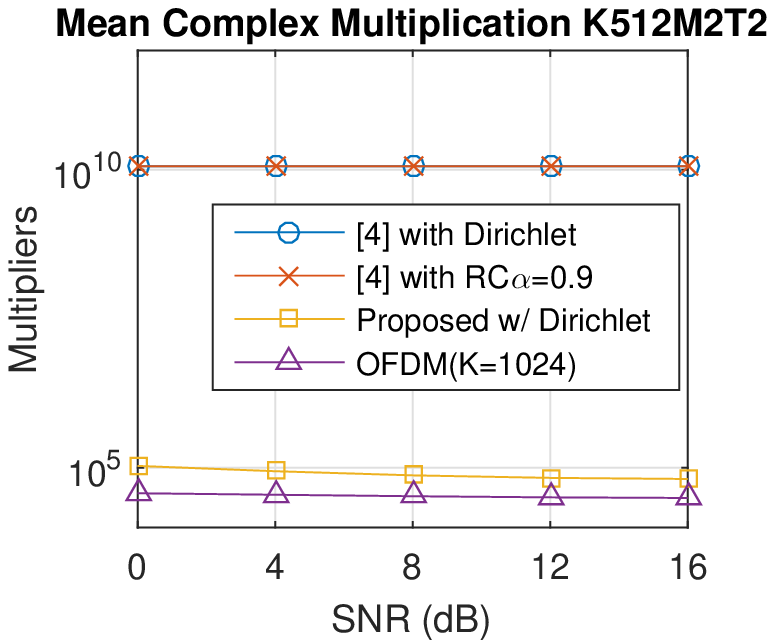}}%
\caption{Performance comparison for $K=512,M=2,T=2$}
\label{fig:K512M2T2_all}
\end{figure}

\section{Conclusions}
\label{sec:conclusion}
In this paper, we propose a non iterative depth-first sphere-decoding (DFSD) detection scheme of MIMO-GFDM spatial multiplexing (SM) by exploiting frequency-domain (FD) decoupling.
We identify a class of prototype filters that enable such FD decoupling and facilitate detection of MIMO-GFDM SM with a significant complexity reduction that was hitherto unrealized.
Our simulation results confirm the considerable complexity drop and a symbol-error-rate (SER) improvement with the proposed scheme.
These results demonstrate that the Dirichlet filter with the proposed scheme more effectively balances SER performance and detection complexity compared with the widely applied RC filter for MIMO-GFDM systems.
In addition, OFDM is demonstrated to be a special case of the proposed scheme, which can be considered a generalization of prototype filter design for FD decoupling.
Future directions include prototype filter design for particular properties and the optimal choice of $M$ in MIMO-GFDM systems.



%




\ifCLASSOPTIONcaptionsoff
  \newpage
\fi



%



\bibliographystyle{IEEEtran}
\bibliography{IEEEabrv,waveform}

%





\end{document}

%% file: newcommands.tex
\newcommand{\vd}{{\bf d}}

\newcommand{\vg}{{\bf g}}
\newcommand{\vh}{{\bf h}}

\newcommand{\vn}{{\bf n}}

\newcommand{\vu}{{\bf u}}

\newcommand{\vy}{{\bf y}}

\newcommand{\zv}{{\bf 0}}

\newcommand{\cA}{{\cal A}}

\newcommand{\cC}{{\cal C}}
\newcommand{\cD}{{\cal D}}

\newcommand{\cN}{{\cal N}}

\newcommand{\cT}{{\cal T}}

\newcommand{\gd}{\delta}

\newcommand{\bN}{{\mathbb{N}}}

\newenvironment{mat}[1]{\left[\begin{array}{#1}}{\end{array}\right]}

\newcommand{\mx}[1]{\begin{mat}{ccccccc}#1\end{mat}}

\renewcommand{\tt}[2]{{\cT_{#1}^{#2}}}

\newcommand{\mA}{{\bf A}}
\newcommand{\mB}{{\bf B}}

\newcommand{\mE}{{\bf E}}
\newcommand{\mF}{{\bf F}}

\newcommand{\mH}{{\bf H}}
\newcommand{\mI}{{\bf I}}

\newcommand{\mO}{{\bf O}}
\newcommand{\mP}{{\bf P}}
\newcommand{\mQ}{{\bf Q}}
\newcommand{\mR}{{\bf R}}

\newcommand{\mU}{{\bf U}}

\newcommand{\mW}{{\bf W}}

\newcommand{\ignore}[1]{}

\DeclareMathOperator{\vect}{vec}

\DeclareMathOperator{\e}{E}

\DeclareMathOperator{\diag}{diag}

\newcommand{\modd}[1]{\left\langle#1\right\rangle}
\newtheorem{theorem}{Theorem}